\documentclass[conference]{IEEEtran}
\IEEEoverridecommandlockouts
\usepackage{hyperref} \hypersetup{colorlinks=true, breaklinks=true, urlcolor=black, linkcolor=black, citecolor=black}
\usepackage{amsmath,amssymb,amsfonts}
\usepackage{graphicx}
\usepackage{subcaption} % for subfigures
\usepackage{textcomp}
\usepackage{xcolor}
\usepackage{xspace}
\usepackage{pifont}
\def\BibTeX{{\rm B\kern-.05em{\sc i\kern-.025em b}\kern-.08em
    T\kern-.1667em\lower.7ex\hbox{E}\kern-.125emX}}
\usepackage{pgf-umlsd}
\usepackage{tikz}
\usepackage{comment}
\usetikzlibrary{shapes,arrows}
\usepackage{ulem}
\usepackage{cite}
\usepackage[para]{footmisc}
\usepackage{algorithmicx}
\usepackage{algpseudocode}
\usepackage{algorithm}
\usepackage{macrosForFunAndProfit}
\newtheorem{assumption}{Assumption}

\newtheorem{lemma}{Lemma}
\newtheorem{theorem}{Theorem}
\newtheorem{proof}{Proof}

\begin{document}

%For the algoritms
\algrenewcommand\algorithmicprocedure{\textbf{upon}}
\algtext*{EndProcedure}% Remove "end Upon" text
\algtext*{EndIf}% Remove "end if" text
\algtext*{EndFor}% Remove "end for" text
%%%%%%%%%%%%%%%%%%%%%%%%%%%% names and comments
\newcommand{\fanta}{FANTASTYC\xspace}
\newcommand{\poia}{{PoA$\&$I}\xspace}
\newcommand{\flaggr}{Aggregator\xspace} \newcommand{\flaggrs}{Aggregators\xspace}
\newcommand{\flclient}{Learner\xspace} \newcommand{\flclients}{Learners\xspace}
\newcommand{\flowner}{Model Owner\xspace} \newcommand{\flowners}{Model Owners\xspace}
\newcommand{\bcval}{Validator\xspace} \newcommand{\bcvals}{Validators\xspace}
\newcommand{\bcclient}{Blockchain Participant\xspace} \newcommand{\bcclients}{Blockchain Participants\xspace}
\newcommand{\dsreplica}{storage replica\xspace} \newcommand{\dsreplicas}{storage replicas\xspace}
\newcommand{\dsclient}{distributed storage client\xspace} \newcommand{\dsclients}{distributed storage clients\xspace}
\newcommand{\adp}[1]{\textcolor{violet}{{Anto: \bf {#1}}}}
\newcommand{\st}[1]{\textcolor{blue}{{Sara:\bf {#1}}}}
\newcommand{\agp}[1]{\textcolor{orange}{{Alvaro: \bf {#1}}}}
\newcommand{\am}[1]{\textcolor{purple}{{Aure: \bf {#1}}}}
\newcommand{\sg}[1]{\textcolor{green}{{Steph: \bf {#1}}}}
\newcommand{\wb}[1]{\textcolor{red}{{William: \bf {#1}}}}
%% Please add your own macro for the comments

%#DOUBLE-BLIND
%\title{Fantastyc: Blockchain-based Federated Learning Made Practical}
\title{Fantastyc: Blockchain-based Federated Learning Made Secure and Practical}
%*\\

\author{
\IEEEauthorblockA{
    \begin{tabular}{c}
        William Boitier, Antonella Del Pozzo, Álvaro García-Pérez, Stephane Gazut, Pierre Jobic, Alexis Lemaire,\\
        Erwan Mahe, Aurelien Mayoue, Maxence Perion, Tuanir Franca Rezende, Deepika Singh, Sara Tucci-Piergiovanni
    \end{tabular}
} %Add Deepika Singh and Pierre Jobic

\IEEEauthorblockA{ 
    \textit{Université Paris-Saclay, CEA, List} \\
    F-91120, Palaiseau, France \\
    \{name.surname@cea.fr\}   
}
}
\begin{comment}
\author{
    \IEEEauthorblockN{Anonymous Authors}
}
\end{comment}
\maketitle

\begin{abstract}
Federated Learning is a decentralized framework that enables multiple clients to collaboratively train a machine learning model under the orchestration of a central server without sharing their local data. The centrality of this framework represents a point of failure which is addressed in literature by blockchain-based federated learning approaches. While ensuring a fully-decentralized solution with traceability, such approaches still face several challenges about integrity, confidentiality and scalability to be practically deployed. In this paper, we propose Fantastyc, a solution designed to address these challenges that have been never met together in the state of the art.

%\adp{TO DO: remove FANTASTYC for the double-blind submission, to be restored if accepted. I use to hashtag \#DOUBLE-BLIND in the comment to spot the modifications}

\end{abstract}

%\begin{IEEEkeywords}
%
%\end{IEEEkeywords}

\section{Introduction}\label{sec:intro}
Federated Learning (FL) has emerged as an important paradigm in modern large-scale machine learning\cite{McMahan2017}. Unlike in traditional centralized learning where models are trained using large datasets stored in a central server, in federated setting, the training data remains distributed over a large number of clients, which may be phones, network sensors, hospitals or companies. A federated model is then trained under the supervision of a central server by aggregating only ephemeral locally-computed updates, thereby ensuring a basic level of privacy by design. However, FL still faces key challenges\cite{Kairouz2019light} to be fully accepted in both institutional and industrial sectors. First, the training process should provide guarantees to ensure its robustness to malicious actors who try to degrade the accuracy of the model (integrity) or retrieve sensitive information about isolated training data (privacy). Second, the federated process should not collapse even if there is a large number of participants to the training step (scalability). Third, the server which orchestrates the federated process is a central player which potentially represents a single point of failure. While large companies or organizations can play this role in some application scenarios, a reliable and powerful central server may not always be available or desirable in more collaborative learning scenarios. 

In the literature, a growing research effort  is dedicated to propose fully decentralized learning through blockchain-based Federated Learning (BC-based FL) \cite{BAFFLE19,FLchain19,BlockFL20,BFEL20, BlockFLA21, Biscotti21, ChainFL20, VFChain22}. Blockchains offer a secure, shared ledger among users, where information is stored and verified by a peer-to-peer network of nodes. The peer-to-peer network adds new information to the ledger thanks to a consensus mechanism tolerant to malicious participants. 
As a result, blockchains provide a means to eliminate the need for trust in a central server. Indeed, the blockchain could be used to validate and store updates and models to secure the FL process. Additionally, blockchains incorporate native incentive mechanisms. These mechanisms encourage users to participate honestly in the system, leveraging  the audit trail of users participation registered in the blockchain. %\adp{The more honest users there are, the stronger the defense against malicious actors.}

To the best of our knowledge, however, current BC-based FL solutions do not address all the challenges as a whole, which makes them not fully secure or practical. A fundamental issue lies in the scalability of blockchains, particularly concerning the latency required to register transactions, a factor influenced by the size of the registered data and the number of clients. In this paper, our objective is to develop a BC-based FL system that ensures both model integrity with proven guarantees, addresses privacy concerns, and maintains practicality even with a large number of clients. Our goal is to implement realistic model training scenarios, akin to those used in classical benchmarking for FL image classification tasks (such as Federated Extended MNIST), involving thousands of clients and models of several megabytes in size. 

% Talk about the decoupled architecture, with clients and servers and the blockchain.  
% We have $n_c$ clients and $n_s$ servers. 
% \adp{develop the concept that the blockchain is just a ledger, that clients and servers can write and read. The blockchain is decoupled from the federated learning application for scalability purposes, as in cosmos. In this was you can have also native sharding, with validators dedicated to each applications. For the practicality} 

To this end, we propose Fantastyc. In our design, we leverage blockchain technology with a specific emphasis on reducing interactions with it, all the while capitalizing on its inherent advantages.   To enhance efficiency, especially given the resource-intensive nature of validating updates and models, we offload the validation process -- meant to be handled by the blockchain -- to a dedicated set of servers. These servers handle aggregation and generate validity proofs, ensuring the integrity of the FL process. Clients can efficiently verify model integrity by examining these validity proofs. Moreover, we tackle the challenge of managing large-sized data that cannot be directly stored on-chain. Our solution involves employing a fault-tolerant key-value store for registering client updates and models as follows.  In each round of the FL process, clients updates are gathered by servers, stored in the storage and cryptographically anchored in the blockchain alongside validity proofs, to record client participation. By reading the updates anchored in the blockchain, each server retrieves an identical set of client updates before aggregating them and produce the next model. %
Fig. \ref{fig:intro} illustrates these ideas, highlighting the additional components we use in our architecture with respect to a centralized, non-Byzantine tolerant solution. The figure focuses on the exchange of data in a single FL iteration. In the centralized version, the clients performing the local training send their updates (the gradients) to the server for that iteration. The server performs aggregation and sends back the aggregated model to clients. In our solution, all participants register in the system via the blockchain before engaging in the learning process. The aggregation is then handled by a set of participants who register to play the role of servers (where a minority can be  Byzantine). Servers perform robust aggregation \cite{Yin2018} to filter out Byzantine client contributions while engaging in a protocol to certify the validity of the obtained aggregated model to send back to clients. At the conclusion of this protocol, a validity proof is generated. Both clients' updates and aggregated models are stored in a replicated storage, while their anchors (i.e., cryptographic hashes) are stored in the blockchain, along with the proofs. Interestingly, servers also write via the proof the contributions from the clients they retained for the aggregation, establishing a trail of client contributions to be exploited for rewarding.

\begin{figure}
    \centering
    \begin{subfigure}[b]{0.4\textwidth}
        \centering
       \includegraphics[width=.8\linewidth]{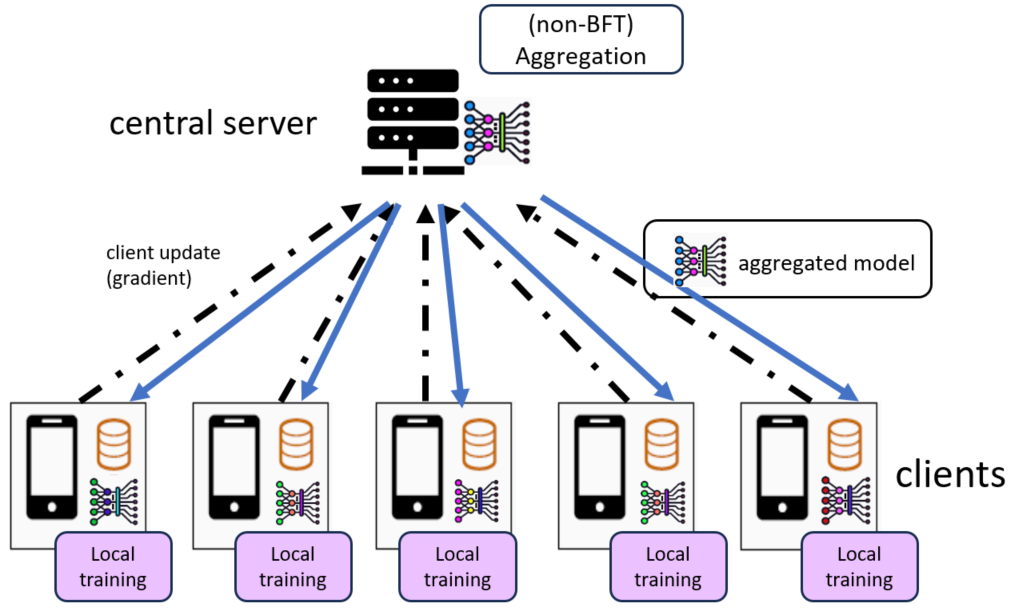}
        \caption{Federated Learning with a central server.}
        \label{fig:sub-intro1}
    \end{subfigure}
    \hfill
    \begin{subfigure}[b]{0.4\textwidth}
        \centering
        \includegraphics[width=1\linewidth]{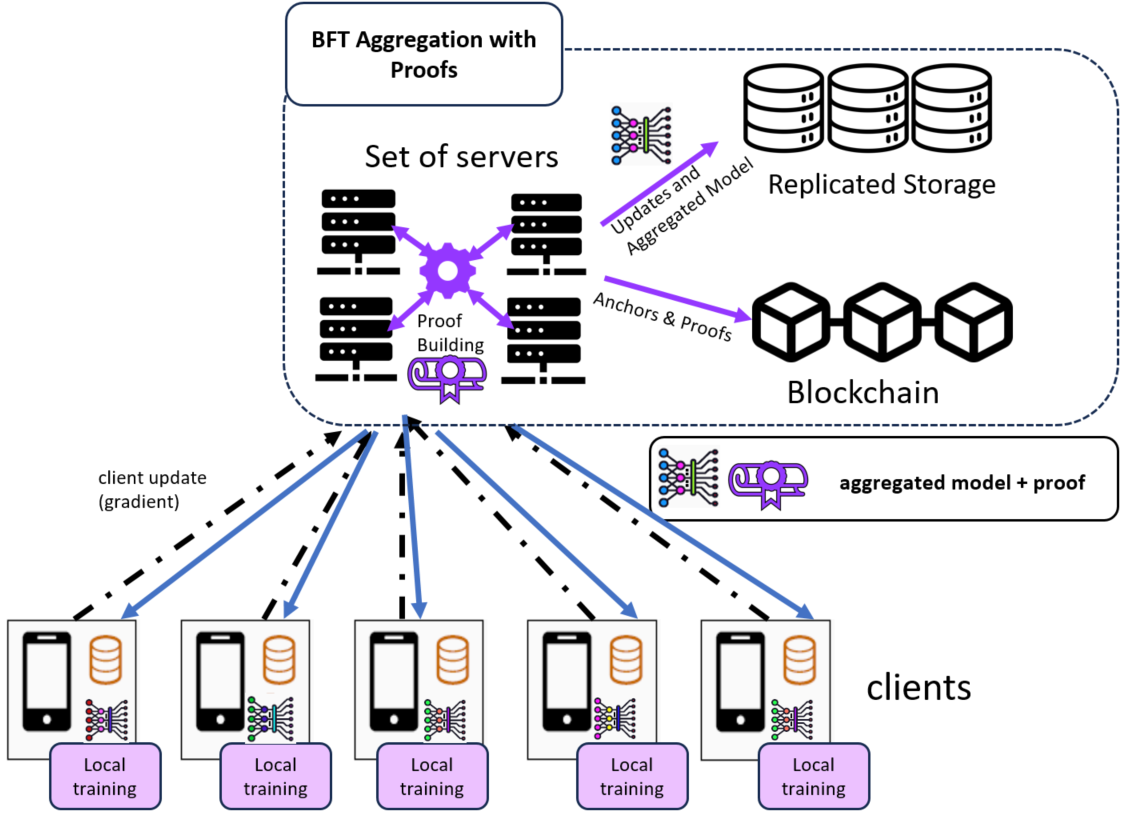}
        \caption{Federated Learning with a BFT set of servers.}
        \label{fig:sub-intro2}
    \end{subfigure}
    \caption{Difference Federated Learning (a) and our solution (b).}
    \label{fig:intro}
\end{figure}

We also employ a light-weight confidentiality technique to face the risk of data leakage. The method makes it difficult for attackers to recover the personal data from users' dataset. 
Notably, this module does not add any communication overhead while only suffering small accuracy loss.
%The main characteristic of this module is that it does not add any communication overhead while only suffering small accuracy loss.  

Overall, our solution achieves the following properties: 
\begin{itemize}
\item \textbf{Robustness}: 
Our solution exhibits robustness against Byzantine attacks \cite{Lamport1982} from both clients and servers, safeguarding the accuracy of the federated model against adversarial attempts to compromise its integrity.% (poisoning attack). 

\item \textbf{Enhanced Byzantine Tolerance}: 
Our solution guarantees model integrity and availability by assigning the blockchain the sole responsibility of managing clients' participation and contributions.  This approach allows the set of servers to maintain an optimized fault-tolerant threshold,  requiring a majority of honest servers rather than the conventional two-thirds \cite{Lamport1982}, making it a more secure option. 

\item \textbf{Efficient Resource Use and Proofs}: Our solution optimizes blockchain usage by storing updates and models within an external distributed storage connected to the blockchain. In this setup, the blockchain exclusively holds a cryptographic fingerprint of the data (the anchor), accompanied by a validity proof. This proof certifies the availability and integrity of the data within the storage. This approach ensures that clients can always  retrieve computed models (availability) providing them with a guarantee of retrieving a trustworthy model (integrity). It is important to note that prevalent architectures in the blockchain space, such as those utilizing external storage systems like IPFS, lack computational integrity even when they offer mechanisms for making redundant copies (proof-of-replication \cite{fileCoin}). Indeed, the primary goal of these architectures is to offer tamper-proof guarantees: clients can only ascertain whether someone has tampered with their own data, without ensuring that the data are the output of correct computations.

\item \textbf{Confidentiality}: Our solution limits the risks of information leakage about
the client's data with little loss of accuracy and without adding any communication overhead. Each client has to complete a simple data processing routine before participating in the training process that consists in mixing his own data samples and applying a random sign-flipping pattern.

%\item \textbf{Auditable participation}: Our solution enables the verifiable tracing of client participation, a critical feature for implementing reward mechanisms and mitigating free-riding\footnote{Definition of specific incentive models is out of the scope of the paper.}.

\item \textbf{Scalability}: Our solution confines the model aggregation to a small number of servers -- outside the blockchain. This ensures a blockchain latency independent of the model size. Additionally, a verifiable random selection process allows us to efficiently poll a subset of clients in different rounds, enabling scalability to accommodate the participation of a large number of clients. Notably, we demonstrate the capability to handle realistic FL processes in a geo-replicated deployment, for the first time for a blockchain-based FL solution.
\end{itemize}

The rest of the paper is organized as follows. Section \ref{sec:background} introduces FL and Blockchain. Section \ref{sec:stateoftheart} discusses the current state of the art of BC-based FL solutions and their limitations. Section \ref{sec:designchoices} delves into design choices of the proposed solution that enhance robustness, Byzantine tolerance, efficient resource utilization and confidentiality. Section \ref{sec:workflow}  presents algorithms and correctness arguments. Section \ref{sec:eval} discusses the implementation and evaluation of the solution in a geo-replicated distributed system. Section \ref{sec:conclusion} concludes the work.

\section{Background and Preliminaries}\label{sec:background}
\subsection{Federated Learning}
In Federated learning (FL), we consider a set of $C$ clients with $n_c$ data that participate in the training of a collaborative model $w \in \Re^d$ under the orchestration of a central server but without sharing their local data due to privacy issues. More formally, if we define $f_c:\Re^d\rightarrow\Re$ as the local objective function for a client $c$, the goal is to minimize the overall objective function $f$ considering the isolated datasets:
\begin{equation}
  \min_{w} f(w) = \min_{w} \sum_{c=1}^C \frac{n_c}{N} f_c(w)\hspace*{0.2cm}\text{where}\hspace*{0.2cm}
  N=\sum_{c=1}^C n_c
\label{eq:loss}
\end{equation}
The FL process consists of a succession of rounds. At each round $t$, the server selects a subset $C_t$ of $K \leq C$ clients that take part in training and sends them the current global model $w_t$. Then, each selected client $k$ uses its dataset to minimize his proper objective function $f_k$ with stochastic gradient descent (SGD), a number of epochs $E$ and a mini-batch size $B$ to update parameters $w_t^k$. Each client $k$ then sends updates $g_t^k = w_t - w_t^k$ to the server who averages the received parameters to obtain a new global model $w_{t+1}$. The aggregation rule depicted by Eq.\ref{eq:fedavg} is known as FedAvg.
\begin{equation}
  w_{t+1}=w_{t}-\sum_{k \in C_t} \frac{n_k}{n} g_t^k\hspace*{0.2cm}\text{where}\hspace*{0.2cm}
  n=\sum_{k \in C_t} n_k
\label{eq:fedavg}
\end{equation}

\subsection{Blockchain}\label{sec:prelim-blockchain}
Blockchain is a decentralized immutable ledger, which consists of a sequence of linked blocks, each containing a set of transactions. Transactions allow blockchain users to immutably record information, which constitutes a common knowledge shared among blockchain users. In particular, this information is transparent, tamper-resistant, and totally ordered. At the core of each blockchain, the consensus mechanism plays a pivotal role in appending new blocks without breaking the chain structure. In our context, we focus on blockchains based on BFT-Consensus, such as Cosmos \cite{cosmos} and Tezos \cite{tezos}, which provide immediate finality, ensuring that once a block is added, it cannot be revoked.

Most blockchains, like Tezos and Ethereum \cite{Buterin2013}, feature smart contracts \cite{szabo} for programmability. Smart contracts are self-executing scripts tied to specific blockchains, deployed and executed through transactions, placing both the generic applications and the blockchain logic on the same layer. In contrast, Cosmos introduces the ABCI (Application BlockChain Interface), separating the blockchain from application logic. With Cosmos ABCI, the blockchain only records transactions, while applications independently process and update their state, enhancing modularity and deployment flexibility.

While blockchain offers security and decentralization, it also faces scalability limitations, particularly in transactions per second and communication latency. For instance, in Tendermint \cite{tendermint} (now CometBFT \cite{abci}), the underlying consensus mechanism for Cosmos, the average throughput is approximately 430 transactions per second with a latency of 3.45 seconds in the presence of 128 correct nodes \cite{tendermintSRDS}. 
These constraints stem from the consensus mechanism and the distributed nature of blockchain, resulting in substantial memory costs as each node maintains an updated replica.
Therefore, blockchain usage must be approached with frugality, emphasizing the efficient and careful use of resources to ensure sustainability and optimize performance.
%For such a reason, in our solution the blockchain is employed only for the core information that needs to inherit the blockchain properties of transparency, temper-resistance and ordering.
%\adp{ref to Tendermint SRDS, to provide some information about the transaction throught and the cost of the underlining consenso. e pure il fatto che e' un struttura data condivisa da tutti i nodi, quindi la vogliamo leggera.}
%\subsection{Blockchain-based Federated Learning}

\section{State of the Art}\label{sec:stateoftheart}

BlockFL \cite{BlockFL20} is a pioneering work that introduces a conceptual architecture where clients record their local updates in a public (proof-of-work) blockchain and perform local aggregation as soon as a sufficient number of local updates are written on the blockchain. 
This approach might be suited for small-sized models that can fit within blocks, even if the paper does not present a concrete solution and evaluation.  Furthermore, the architecture relies on each client being associated with a single miner for sending and downloading blocks. However, a malicious miner could potentially pose a threat by refusing to include a client's local update in the blockchain or by remaining silent, causing availability issues. Finally, in the proposed architecture, miners are expected to validate the correctness of client updates by confirming that each update has been executed across the entire client dataset, because a client might send fake updates without incurring the cost of training. This verification relies on both miners having knowledge of each client's data size and the presence of a trusted execution environment at each client, ensuring the integrity of the clients' contributions. 
%Because the effectiveness of a trusted execution environment in ensuring the correctness of the updates depends on whether the entire training process can be conducted within such an environment, the lack of detailed information about how the trusted execution environment is utilized creates a lack of clarity about the assumed threat model for clients and the type of attacks under consideration.
%However, the precise mechanisms employed within these trusted execution environments for update computing are not explicitly described. Because the effectiveness of a trusted execution environment in ensuring the honesty of clients depends on whether the entire training process can be conducted within such an environment, the lack of detailed information about how the trusted execution environment is utilized creates a lack of clarity about the assumed threat model for clients and the type of attacks under consideration.
%Regarding incentives, 

FLChain \cite{FLchain19} shares similarities with BlockFL in its use of blockchains to record local updates. However, it diverges from BlockFL by introducing an additional cryptographic layer to encrypt the local updates stored in the blockchain. Aggregation is then executed by (a subset of) blockchain nodes on these encrypted updates, and the resulting aggregated model is stored in the blockchain. Subsequently, aggregated models are collectively decrypted by clients, with each client holding the ability to decrypt only a share of each aggregated model. The approach also includes incentive mechanisms designed to motivate participants to provide uncorrupted shares. Much like BlockFL, this approach does not explicitly address concerns related to Byzantine clients conducting poisoning attacks. Also, due to the substantial utilization of the blockchain and the collaborative decrypting algorithm, it is better suited for smaller-sized models and a few number of clients.

VFChain\cite{VFChain22} shares architectural concepts similar to FLChain and BlockFL. However, it distinguishes itself by providing finely-tuned mechanisms for the dynamic election of an aggregation committee. Nonetheless, the storage of aggregated models in the blockchain is maintained, and experimental evaluations have been conducted with a participant count not exceeding 10. It is important to note that in this system, clients are regarded as honest but curious, while Byzantine behavior is attributed solely to aggregators, with a tolerance for up to one-third of them. In BFLC\cite{BFLC21}, a similar architecture with a committee is used for efficiency, whereas a robust aggregation function is implemented by committee members to defend the system from Byzantine clients. 

BFEL\cite{BFEL20}, in comparison to the aforementioned approaches, introduces a more sophisticated architectural design. It combines a public blockchain with private blockchains, with each private blockchain dedicated to a specific FL task. This approach optimizes storage efficiency in multi-task scenarios and enhances privacy. Additionally, to alleviate communication overhead resulting from the exchange of large models, BFEL applies a gradient compression method. In this system, a committee of verifiers plays a crucial role in filtering out potentially poisoned models submitted by clients. However, it is worth noting that the verification method relies on each verifier using a testing dataset to evaluate the updates, instead of employing more robust aggregation techniques that might offer greater reliability. Experimental evaluations have been conducted with two parallel images classification tasks, in a local machine, without exceeding 22 clients. 

BlockFLA\cite{BlockFLA21} also introduces a hybrid architecture, where a private blockchain is integrated with a public blockchain. In this setup, model updates are entirely stored within the private blockchain and are subsequently linked to the public blockchain for accountability via a smart contract. While the paper mentions the existence of a verification function within the smart contract, it does not provide specific details regarding its functioning.

Finally,  Biscotti \cite{Biscotti21} refines and elaborates on principles stipulated by FLChain\cite{FLchain19} by introducing a three-role structure: noisers create unbiased noise for clients to mask local updates, verifiers endorse masked updates through robust filtering, and aggregators confidentially aggregate the endorsed updates. Aggregators work with encrypted data to ensure privacy. The encrypted aggregated model is stored on the blockchain within commitments to endorsed updates. Unlike our approach, Biscotti, by storing the model on the blockchain, can only handle small models (models with thousands of parameters as opposed to millions in our approach). Additionally, there's uncertainty about the purpose of the noise particularly considering the authors' statement that they do not utilize it at the end due to significant utility loss. Without noise, verifiers access unmasked client updates, posing a privacy risk if even a single verifier is compromised. Therefore, finding the right trade-offs between noise and utility is essential to ensure privacy while using this architecture.

It is also worth to mention some other research efforts, such as those in \cite{BAFFLE19} and \cite{ChainFL20}, that have delved into solutions within the Ethereum blockchain framework to replace the model-aggregating server with smart contracts. However, deploying these smart contracts on public blockchain systems can incur significant costs, mainly due to miner fees. These costs are associated with the computational workload and the substantial storage requirements resulting from full replication on every blockchain node.  As a result, these approaches are typically more suitable for handling simpler models and tasks.
%#DOUBLE-BLIND
%Fantastyc 

Our solution distinguishes itself from existing state-of-the-art solutions by minimizing blockchain utilization through a novel off-chain Byzantine-tolerant protocol leveraging distributed storage servers, ensuring the integrity and availability of off-chain data and computations without depleting system resources. Moreover, 
%#DOUBLE-BLIND
%Fantastyc 
our construction seamlessly adapts its protocol to limit the risk of information leakage without introducing any additional overhead while integrating secure guarantees to remain robust in the presence of both Byzantine clients and servers. Notably, it offers an enhanced level of Byzantine tolerance compared to existing solutions, requiring a majority of honest servers rather than the conventional two-thirds, making it a more secure option. 
%#DOUBLE-BLIND
%Fantastyc 
%Our solution also facilitates auditable client participation and enables verifiable implementation of client selection and incentive mechanisms thanks to the blockchain. Lastly, because of its modular architecture, 
%#DOUBLE-BLIND
%Fantastyc 
Interesting, because of its  modular architecture, our solution offers as well the flexibility to be deployed as a private, public, or hybrid blockchain while allowing users to initialize learning tasks with any type of aggregation, selection, and rewarding function. %and to execute any type of aggregation function published by the model owner.  

\section{Design Choices and Key Ideas}\label{sec:designchoices}

\subsection{Robustness and Enhanced Byzantine Tolerance}\label{subsec:robustness}
Every distributed process is vulnerable to misbehaving nodes that can disrupt the learning procedure. Indeed, arbitrary failures due to hardware defaults may occur or some actors may also try to deliberately corrupt the model. Following the terminology of distributed computing, these nodes are called Byzantine\cite{Lamport1982}. We refer to the other nodes as correct or honest. The robustness of a BC-based FL process is thus ensured when model integrity is preserved despite the presence of both Byzantine clients and servers.

\subsubsection{Byzantine Clients}\label{par:fedavg}
In a BC-based FL process, the servers cannot check whether the updates they receive from the clients are coherent with licit training data, which makes the learning procedure vulnerable to update manipulation. Considering the standard aggregation rule FedAvg (eq.\ref{eq:fedavg}), it is very easy for Byzantine clients to hinder the model's convergence by sending erroneous updates such as $-\lambda_{boost} g_t^{k}$ instead of $g_t^{k}$ (where $\lambda_{boost} \in \left[ 1 , +\infty \right[$ is used to boost the strength of the attack). To mitigate the negative impact of Byzantine nodes, it is common to use a non-linear robust aggregator such as the coordinate-wise median\cite{Yin2018}. Considering $f_c$ Byzantine clients, this Byzantine-tolerant aggregation rule requires at least $2f_c+1$ participants to each round of training.   

\subsubsection{Byzantine Servers}
In the context of ensuring robustness against Byzantine attacks, a critical challenge lies in guaranteeing that clients do not retrieve poisoned models from potentially malicious servers. This challenge can be faced harnessing the power of blockchain technology to ensure the integrity of the models collected by clients.
To illustrate this approach, let us envision the blockchain as a ledger -- an immutable, ordered list of records accessible for both clients and servers to write and read. We present two distinct strategies:

(i) In the first strategy, each server collects a (possibly different) set of client updates, performs robust aggregation, and subsequently records the resulting model on the blockchain. It is worth noting that due to network delays and message losses, each server may receive a different subset of updates. In this scenario, clients retrieve from the ledger, in each round, the aggregated models from a count of servers not exceeding a specified fault threshold, denoted as $n_s-f_s$, where $f_s$ accounts for potentially silent or slow servers. Clients then engage in robust aggregation to derive the model for the next round. %\adp{Should we stress that clients cannot select the first model register on chain because it could be generated by a Byzantine server?}\st{why? they do robust aggregation locally, they retrieve n-f models.}\adp{yes, then cannot do otherwise because of byzantine servers. right?}\st{right}

(ii) In the second strategy, clients take on the responsibility of directly writing their updates to the blockchain. For each round, servers await the writing in the ledger of updates from a predefined number of clients. %, not exceeding a designated fault threshold denoted as $n_c-f_c$, where $f_c$ accounts for possibly silent or slow clients. 
In this scenario read updates are the same set at each server. Subsequently, servers engage in robust aggregation and synchronize to cast their votes in favor of models that fall within an epsilon range\footnote{It is noteworthy that, owing to the computational intricacies of floating-point calculations, models generated from identical inputs may exhibit slight differences. Empirical observations, using models with 500,000 parameters, have shown disparities up to three decimal places, underscoring the level of precision involved.\label{footnote:floating}}. When a model accumulates at least $f_s+1$ votes, it signifies that at least one honest server has endorsed it, rendering it safe for use by any client. A client retrieves then a model with $f_s+1$ server votes. 

Assuming the client count ($n_c$) surpasses that of servers ($n_s$), the first strategy is better in terms of ledger storage space used. However, it necessitates a higher fault tolerance threshold, requiring $n_s=3f_s+1$ (to account for slow servers in the round and malicious speaking servers). Additionally, it mandates clients to perform robust aggregation on their end. 
The second approach requires a fault tolerance threshold of only $n_s=2f_s+1$ (to account for slow servers in the round and to guarantee that at least $f_s+1$ votes are cast for a model) and more importantly it spares clients from the computational overhead of robust aggregation. We opt then for the second approach to avoid overload clients with robust aggregation computation and because the better tolerance threshold. To maintain practicality, we also  prevent clients from directly writing to the blockchain. %\adp{This also protect clients to disclose their contribution to a particular round.} 
Instead, we propose a byzantine tolerant protocol where servers gather updates and anchor the collected set in the blockchain with the help of an external fault-tolerant distributed storage, as detailed next. %in the following section.

\subsection{Efficient Resource Use by Proof  of Availability $\&$  Integrity and Optimally Replicated Fault-tolerant Key-value Store}\label{subsec:efficient-use}
To efficiently utilize the blockchain, we propose a Byzantine-tolerant protocol that allows for the reliable storage of both client updates and aggregated models in a fault-tolerant key-value store. Because we need to achieve consensus only on the client updates set, the only information recorded in the blockchain is the cryptographic hook of client updates contained in the store, referred to as Proof of Availability $\&$ Integrity (PoA$\&$I). The PoA$\&$I is generated each time a value (update or model) is consistently and reliably replicated at enough key-value store replicas. This proof guarantees clients and servers to always retrieve the value (\textit{availability}) and ensures that the value is the result of a correct computation (\textit{computational integrity}) -- for instance the correct model of a robust aggregation. In the blockchain, consensus is reached on the set of proofs uniquely associated with correct client updates retrievable by servers in the store.
%\st{Note. validity==computational integrity. Integrity includes computational integrity.}

At high level, the protocol to generate the PoA$\&$I works as follows: as soon as a server gathers a value  $v$ for which a PoA$\&$I is required, the server will verify if the value is valid. If valid, it stores the pair $(hash(v),v)$ in the key-value store where the hash of $v$ is used as key, and then broadcasts a signed message with the hash of the valid value, called the \textit{hash-key}, to the other servers. As soon as a server receives a $f_s + 1$ quorum of signed valid hash-keys from other servers, it produces a PoA$\&$I\footnote{ PoA$\&$I is similar in spirit to the Proof of Availability and Retrieval \cite{poar}, which, however, does not target computational integrity.  Indeed, \cite{poar} describes a mechanism to produce a proof of availability to implement a secure pull communication mechanism for blocks inside blockchains.}.  In the first step of the protocol, the verification performed to validate values depends on the type of value: in the case $v$ is a client update, the server verifies if the update comes from the set of clients selected for the current round;  if $v$ is a model aggregated by other servers, then the value is considered valid if the received model and the locally computed model are not significantly different, typically within an epsilon distance.

%Moreover, this allows any process retrieving a value from the storage to verify that the obtained value is indeed the requested one, which allows us to implement the distributed storage requiring only one correct replicas, while all the others can be faulty. 

 %but aims mainly at improving the performance by decoupling the block diffusion from the next block consensus, which is an expensive part in terms of latency. In the same spirit, by agreeing on information PoIA, we remove the need to verify the information integrity and availability from the consensus critical path.

To ensure persistence and counteract potential attackers attempting to erase or modify off-chain data from the storage, the solution mandates the hash-key always to be the cryptographic hash of the corresponding value, if it is not the case, the store reject the value%. This approach prevents Byzantine nodes from overwriting or delete values associated to a key, 
as it is not possible to associate to a key a value whose hash differs from the key itself.

\subsubsection{One blockchain transaction per round}
To minimize the number of blockchain transactions, instead of recording each client update's  PoA$\&$I on the blockchain immediately, servers wait to generate a sufficient number of client updates' PoA$\&$I. When this threshold is reached, they write a set of updates' PoA$\&$I  as a candidate input for the next robust aggregation on the blockchain by sending a transaction. Since transactions on the blockchain are totally ordered, the first set of updates, each with a valid PoA$\&$I, recorded on-chain is the one utilized for robust aggregation by all correct servers. It is important to note that after agreeing on a set of PoA$\&$Is, servers perform robust aggregation and generate a PoA$\&$I for valid models computed and stored. The PoA$\&$I of the aggregated model  is then sent to clients, who can retrieve models directly from the store without interacting with the blockchain.

%%% up to here we have the minimal information to understand the PoIA.

%\adp{I think we can remove this paragraph:}
%It is important to notice that, at each iteration, there are multiple updates sets and multiple aggregated models, both having a PoIA. While this helps to distinguish valid from invalid ones, this does not help to have a common agreement on the updates set or model to consider. As discussed before, different clients can consider different aggregated models, in other words, they can disagree on the aggregated model to consider as long as it has a PoIA which guarantees its integrity. However, our solution does not provide the same degree of freedom on the updates set, indeed all servers must agree on the same set to produce aggregated models sufficiently close to each other (more details in Section \ref{sec:non-det}) whose agreement is achieved through the blockchain.

\subsubsection{Optimal replication}

%% I think that this paragraph can be removed
%Operationally, each server sends to the blockchain a candidate updates set for the current round, the first valid one on-chain is considered as input set for the robust aggregation. \st{je ne comprends pas cette phrase:} Let us stress that the robust aggregation operation is decoupled from the transaction handling, the latter involves the PoIA verification and the application (FL) state update.  Which allows to decouple  on-chain information handling from the federated learning part. \adp{I wanted to say that we do not mix consensus and FL, the aggregation is not performed in the delivery tx, but it can be done in parallel}
It is worth noting that  our approach decouples the ordering (blockchain side), the integrity (server side) and availability (storage side) of the information. 
This decoupling allows to implement those properties employing the optimal number of nodes. As discussed in  Section \ref{subsec:robustness}, the number of servers is at least $2f_s+1$. 
Concerning the key-value store, $f_r+1$ replicas suffices, where $f_r$ is the upper bound on the number of Byzantine replicas and each replica is replicating the content of the key-value store. 
Thanks to the  PoA$\&$I  associated to the hash-key, a client querying the store can assess the validity of a returned value by checking if the retrieved value, once hashed, is equal to the required hash-key. As a consequence, to get the value, it is sufficient to get a correct answer from at most one correct storage replica, which exists, given the existence of the PoA$\&$I.
%The pseudo-code in Algorithms \ref{algo:server}-\ref{algo:read} formalizes the previous discussion.

Let us stress that, only updates associated with a PoA$\&$I are considered part of a valid updates set. This ensures that any set recorded in the blockchain, if valid, contains updates available in the storage and produced by selected clients. This mechanism provides certainty to a server; if they miss some updates necessary to aggregate the next model, they are assured they can retrieve them in the storage.  %\st{also clients are ensured to find the model and a valid one. I think this can be said here by combining the paragraph below in some manner.}
A similar approach applies when servers interact with clients. Indeed, once servers compute the next aggregated model, they store it in the distributed storage, and create a proof before sending it directly to clients. This approach provides clients with verifiable evidence of integrity and availability, mitigating concerns related to computational manipulation. %Clients are only required to know the set of servers to be able to verify the proof.
Detailed workflows and algorithms can be found in Section \ref{sec:workflow}.

%Subsequently, servers share $PoIA(w)$ with the clients. 

\subsection{Confidentiality}
By design, FL intrinsically protects the data stored on each device by sharing model updates (\textit{i.e.} gradients) instead of the original data. While sharing gradients was thought to leak significantly less information about the clients compared to the raw data, recent papers developed a ``gradient inversion attack" by which honest-but-curious servers are able to reconstruct parts of the client's private data\cite{DLGZhu2019,Geiping2020}. Homomorphic encryption is the gold standard for a privacy-preserving solution but its large computational demand still makes it not compatible with scalability despite recent improvements\cite{choffrut2023practical}. Based on literature, we propose countermeasures to limit the risks of information leakage. We advise clients who want to participate in the training process to apply InstaHide\cite{pmlr-v119-huang20i} data augmentation to their dataset. In a nutshell, InstaHide has two key components. The first step consists in building virtual training samples in the following way: to encode an example $x \in \Re^d$, a client picks $s-1$ other samples from his private dataset and $s$ random nonnegative coefficients $\{\lambda_j\}_{j=1}^{s}$ that sum to~1. A composite (one-hot encoding) label is also created using the same set of mixing coefficients:
%\begin{equation}
%\begin{split}
%\widetilde{x}=\lambda_1 x+\sum_{j=2}^{s}{\lambda_j x_j}\\
%\widetilde{y}=\lambda_1 y+\sum_{j=2}^{s}{\lambda_j y_j}
%\end{split}
%\label{eq:mixup}
%\end{equation}
\begin{equation}
x=\lambda_1 x+\sum_{j=2}^{s}{\lambda_j x_j} \hspace*{0.2cm}\text{and}\hspace*{0.2cm} y=\lambda_1 y+\sum_{j=2}^{s}{\lambda_j y_j}
\label{eq:mixup}
\end{equation}
As this step permits to train a neural network on linear combinations of examples instead of raw data, an attacker could only retrieve the virtual examples. The capacity of inferring the raw data from the virtual ones requires the knowledge of the mixing coefficients but this information seems hard to obtain in a realistic scenario.
However, InstaHide adds an another layer of security that consists in picking a random sign-flipping pattern $\sigma \in \{-1,1\}^d$ and outputting the encryption $\widetilde{x}=\sigma\circ(\lambda_1 x+\sum_{j=2}^{s}{\lambda_j x_j})$ where $\circ$ is coordinate-wise multiplication of vectors.

\section{System Detailed Design}\label{sec:workflow}
In this section we  first introduce the system model and a formalisation of the \poia,  before presenting the workflow of the FL process, which includes algorithms to build the  \poia (Section~\ref{subsec:workflow}) and methods to manage the non-determinism of the computations (Section~\ref{sec:non-det}). %Finally, we defer the correctness arguments to Appendix~\ref{app:proofs}. 
We finally present correctness arguments of the solution (Section~\ref{sec:proofs}).

\subsubsection{System model} %Our system is composed of a set of processes. Each node can take on one of the following roles: 
Our system is composed of a finite set of processes\footnote{The term process is used here to denote a computer thread, distinct from the usage of FL process referring to the Federated Learning process. Both terms are employed within their respective contexts to avoid any confusion.}: clients, servers, storage replicas,  Model Owner (MO), all having access to a blockchain. The blockchain, an immutable, ordered list of records, allows processes to communicate and share a common knowledge upon the information securely recorded on it. Clients, servers and storage replicas can also directly communicate. 
More formally, we assume that the system comprises $2f_s + 1$ servers ${S}$ and $f_r + 1$ storage replicas ${R}$, where $f_s$ and $f_r$ represent the upper bounds on Byzantine faulty processes (processes that exhibit an arbitrary behaviour). We assume model owners to be correct. Additionally, we rely on the reliable communication assumption that if 
%\begin{assumption}\label{ass:communication}
    a correct process $p$ sends a message $m$ to another correct process $q$, then eventually $q$ receives $m$. Moreover, if a correct process sends a transaction $t$ to the blockchain, then $t$ is eventually totally ordered in the blockchain. All messages are signed by the sender.
%\end{assumption}
%The FL process starts with the model owner that sets up the FL task and awaits its completion to collect the final model and automatically reward the clients that trained it. 
% In the following (Subsection \ref{subsec:workflow}), we present a high level view of the solution workflow for an FL task $fl_i$, from the moment a model owner inputs the model $w_0$ to be trained, to the moment the model owner retrieves the trained model $w_{t_{fin}}$, after $t_{fin}$ rounds. 
% Algorithms \ref{algo:server}-\ref{algo:read} provide details on the core part of the process (in the pseudo-code we do not explicit the particular $fl_i$ considered, as it can be easily generalized to handle multiple FL tasks in parallel). In particular, Algorithm~\ref{algo:server} details the server specification as described from step 4 to step 7 in the workflow. 
% It also describes how \poia works. 
% Algorithm \ref{algo:read} specifies how nodes interact with the secure distributed key-value store to get a value once they collect its PoA$\&$I. Given its simplicity, we do not provide the pseudo-code of the replica node, when necessary we describe its behaviour in comments in Algorithms \ref{algo:server} and \ref{algo:read}. 
We also consider that servers use the same pseudo-random function {\sf select\_clients}, that takes as input the blockchain (which contains the list of clients) to verifiably select the client set for the next round.

\subsubsection{\poia} We refer to the \poia as $P(\tagtt,v)$, a set of cryptographic hashes of $v$ signed by $f_s+1$ servers. Let $\text{{\sf locally\_valid}}(\tagtt, v)$ be a function that returns {\sf true} if $v$ passes the computational integrity check relative to $\tagtt$. A correct server signs $h(v)$ if $\text{{\sf locally\_valid}}(\tagtt, v)$ holds and after storing $v$ on the key-value store.
If a process $p$ receives $P(\tagtt,v)$, then $\text{{\sf locally\_valid}}(\tagtt, v)$ holds for at least one correct server\footnote{When {\sf locally\_valid} is deterministic, it holds for all correct servers.}, and eventually, $p$ gets $v$ from the key-value store.

\begin{algorithm}
\scriptsize
\caption{Server i}\label{algo:server}
\begin{algorithmic}[1]
    \State $bc$; \Comment{Blockchain object}
    \State $t \leftarrow 0$; \Comment{Current round}
    \State $m$; \Comment{Minimum number of expected clients' updates to aggregate them}
    \State $\mathcal{P}(G_t)_{cand}\leftarrow \emptyset$; \Comment{Set of updates' PoA$\&$I $P(h(g_t^k))$ candidate for round $t$}
    \State $\mathcal{P}(G_t)_{commit}\leftarrow \emptyset$; \Comment{Set of updates' PoA$\&$I $P(h(g_t^k))$ commit for round $t$}
    \State $\mathcal{C}_t \leftarrow \emptyset$; \Comment{Set of $C_t$ elements, the selected clients for the round $t$}
    \State $\mathcal{\lptt} \leftarrow \emptyset$; \Comment{Set of local proofs for values}
    % \Statex
    % \Procedure{locally computed model}{v}
    % \State $PoIA(h(v_i,t))$={\sf buildPoIA}{($v_i,t, model$});
    % \State {\sf notify}($PoIA(h(v)), model)$ clients;
    % \EndProcedure
\Statex
\Procedure{receive}{$\langle \updtt, g_t^k \rangle$}\Comment{Update $g$ by client $k$ for round $t$}
    \If{$k \in \mathcal{C}_{t}$} %\Comment{For simplicity we do not handle messages from the past or future rounds}
        \State {\sf buildPoA$\&$I}${(\updtt, g_t^k}$);\label{algo-server-update-poia}
        \State $\mathcal{C}_{t} \leftarrow \mathcal{C}_{t} \setminus k$;
    \EndIf
    \EndProcedure
\Statex
\Procedure{event}{$|\mathcal{P}(G_t)_{cand}|\geq m$ for the first time}\label{algo-server-candidateset}
    \State $C_{t+1} \leftarrow {\sf select\_clients(bc)}$; \Comment{Based on the client's JOIN transactions} \label{algo-server-select-clients}
    \State ${\sf buildPoA\&I}(\clientstt, C_{t+1})$;\label{algo-server-clients-poia}
    %\State ${P}(C_{t+1}) \leftarrow {\sf buildPoA\&I}(\clientstt, C_{t+1})$;
    %\State ${\sf sendTx}\langle NEW\_R, t, \mathcal{P}(G_t)_{cand}, C_{t+1}, P(C_{t+1}) \rangle$ to blockchain;
\EndProcedure
\Statex
\Procedure{deliverTx}{$\langle \nrtt, t, \mathcal{P}(G_t)_{cand}, C_{t+1}, P(C_{t+1} \rangle$}
    \If{${\sf verify\_proof}(\mathcal{P}(G_t)_{cand}) \wedge C_{t+1}) \wedge \mathcal{P}(G_t)_{commit} = \emptyset$}
        \State $\mathcal{P}(G_t)_{commit} \leftarrow \mathcal{P}(G_t)_{cand}$
    \EndIf
\EndProcedure
\Statex
\Procedure{event}{$\mathcal{P}(G_t)_{commit}\neq \emptyset$ for the first time}
    \State $G_t \leftarrow {\sf getAll}(\mathcal{P}(G_t)_{commit}\setminus \mathcal{P}(G_t)_{cand})$;
    \State $w_{t+1}^i \leftarrow {\sf robust\_aggregation}(G_t)$;
    \State ${\sf buildPoA\&I}{(\modtt,w_{t+1}^i)}$;
    %\State ${\sf publish}(\modtt, P(w_{t+1}^i))$ to clients;
    %\State t++;
\EndProcedure
\Statex
\Procedure{buildPoA$\&$I}{$\tagtt, v$}:
    \If{{\sf locally\_valid}($\tagtt, v$)}
        \If{$\tagtt \neq \clientstt$} \Comment{Only models and updates are sent to the storage}
            \State {\sf send} {$\langle h(v),v \rangle$} to all $r \in {R}$; \Comment{Each storage replica, upon reception, checks if $h(v)=v$ before storing $v$ with hash-key $h(v)$}
            \If{$\tagtt = \modtt$}
                \State {\sf send} {$\langle \modtt, v \rangle$} to all $j \in {S}$;
            \EndIf
        \EndIf
    \State {\sf send} {$\langle \lptt, \tagtt, h(v)_i\rangle$} to all $j \in {S}$; \Comment{$h(v)_i$ is signed by server $i$}
    \EndIf
\EndProcedure
\Statex
\Procedure{receive}{$\langle \modtt, v \rangle$} for the first time
    \If{{\sf locally\_valid}($\tagtt, v$)}
        \State {\sf send} {$\langle h(v),v \rangle$} to all $r \in {R}$;
        \State {\sf send} {$\langle \lptt, \tagtt, h(v)_i\rangle$} to all $j \in {S}$;
    \EndIf
\EndProcedure
\Statex
\Procedure{receive}{$\langle \lptt, \tagtt, h(v)_j \rangle$}
    \State $\mathcal{\lptt} \leftarrow \{(\lptt, \tagtt, h(v)_j)\}$;
    \If{$\mathcal{\lptt}$  contains $f+1$ $(\lptt, \tagtt, h(v)_*)$}
        \State $P(h(v)) \leftarrow$ $f+1$ $(\lptt, \tagtt, h(v)_*)$;
        \If{$\tagtt=\updtt$}
            \State $\mathcal{P}(G_t)_{cand} \leftarrow \mathcal{P}(G_t)_{cand} \cup P(h(v))$
%\EndIf
        \ElsIf{$\tagtt=\clientstt$}
%\State ${P}(C_{t+1}) \leftarrow {\sf buildPoA\&I}(\clientstt, C_{t+1})$;
            \State ${\sf sendTx}\langle \nrtt, t, \mathcal{P}(G_t)_{cand}, C_{t+1}=v, P(h(v))) \rangle$ to bc;
            \State \Comment{If $t=t_{fin}$, the last round, then rather than $C_t$, this transaction carries the \poia of the rewarding information.}
        \ElsIf{$\tagtt = \modtt$}
%\State $P(w_{t+1}^i)=P(h(v))$;
            \State ${\sf send}(\modtt, t, P(h(v))$ to clients;
            \State $t \leftarrow t + 1$;
        \EndIf
        \EndIf
\EndProcedure
\end{algorithmic}
\end{algorithm}

\begin{algorithm}
\scriptsize
\caption{Read from the key-value store at process i}\label{algo:read}
\begin{algorithmic}[1]
\Procedure{get}{$P(\ell)$}\label{algo-read-get}
\If{${\sf verify\_proof}(P(\ell))$} \Comment{verify if $h(\ell)$ is signed by $f_s+1$ servers}\label{algo-read-check}
\State {\bf return} ${\sf read}(P(\ell))$;
\EndIf
\EndProcedure
\Statex
\Procedure{getAll}{$\mathcal{P}(Set)$} and $ ({\sf verify\_proof}(P(\ell)), \forall \ell \in Set)$:\label{algo-read-getall}
\State $buffer \leftarrow \emptyset$;
\ForAll{$P(\ell) \in \mathcal{P}(Set)$}
\State $buffer \leftarrow buffer \cup {\sf read}(P(\ell))$; \label{algo-read-getall-if}
\EndFor
\State {\bf return} $buffer$;\label{algo-read-getall-return}
\EndProcedure
\Statex
\Procedure{read}{$(P(\ell))$}:\label{algo-read-read}
%\If{${\sf verify\_proof}(P(\ell))$} \Comment{verify if $h(\ell)$ is signed by $f_s+1$ servers}\label{algo-read-check}
\State ${\sf send} {\langle \tt QUERY,(h(\ell)) \rangle}$ to all $r \in {R}$; \label{algo-read-request}
%\State $pending \leftarrow pending \cup (h(\ell))$;
%\EndIf
\State {\bf wait until} {received}{$(\tt REPLY, \ell')$ and $h(\ell)==h(\ell')$} {\bf return} $\ell'$;\label{algo-read-return}
\EndProcedure
\end{algorithmic}
\end{algorithm}

\subsection{Solution workflow and algorithms}\label{subsec:workflow}

We outline the workflow for an FL task, denoted $fl_i$, from the moment the model owner inputs the model for training to the moment the model owner retrieves the trained model.
Algorithms~\ref{algo:server}-\ref{algo:read} detail the core part of the process (pseudo-code does not explicit any particular $fl_i$, as it can be easily generalized to handle multiple FL tasks in parallel). Algorithm~\ref{algo:server} details the server pseudo-code, as described from step 4 to step 7 in the following workflow, and how the \poia works. Algorithm~\ref{algo:read} specifies how processes (clients, servers or model owner) interact with the secure distributed key-value store to get a value once they collect its PoA$\&$I. Due to its simplicity, we omit the pseudo-code for the replica storage process. Instead, we describe its behavior in comments within Algorithms~\ref{algo:server} and~\ref{algo:read} when needed.  \textit{Workflow steps:}
\begin{enumerate}
    \item The process starts with the model owner that sends a transaction to the blockchain to run a new FL task $fl_i$. This transaction carries all the information (\textit{e.g.} the hash-key of the model to train in the store) and $param_{fl_i}$, the parameters necessary for the task execution. The set-up round $t=0$ starts;
    
    \item Once the $fl_i$ task is registered on the blockchain, all clients and servers are aware of it. Each client sends a transaction to the blockchain to state its interest to join (if any). Such a transaction stakes a certain amount of coins\footnote{``Staking" in the context of cryptocurrency refers to the process of actively participating in the operations of a blockchain network. This involves locking up a certain amount of cryptocurrency as collateral to perform these tasks. } for $fl_i$, witnessing the client commitment to participate. When the client's transaction is registered on the blockchain we say that the client is registered for $fl_i$. Servers wait for sufficiently many clients to be registered (the number of awaited clients is set in $param_{fl_i}$).
    
    \item Once there are sufficiently many clients registered for $fl_i$, each server selects deterministically out of them a subset $C_{1}$ to perform the learning for the first round and builds $P(C_{1})$, the PoA$\&$I for it (cf. Algorithm~\ref{algo:server}). Finally, it sends a transaction ${\tt START\_TASK}$ to announce to the clients in $C_{1}$ that they are selected for the first round. Once this transaction is registered on the Blockchain the set-up round is over and the computation moves to the first learning round $t=1$. 
    
    \item \label{loop:first-step} Each client $k$ in $C_{t}$, retrieves a previously aggregated model associated to a \poia, (cf. Algorithm~\ref{algo:read}), performs the learning on its local data to produce the update $g_t^k$, signs it, and sends it directly to all servers;
    
    \item Each server upon reception of $g_t^k$ for the first time, builds a PoA$\&$I for it (cf. Algorithm~\ref{algo:server}).

    \item When a server collects sufficiently many updates's PoA$\&$I for round $t$ (as determined in $param_{fl_i}$), it has a candidate update set, $\mathcal{P}(G_t)_{cand}$. It computes $C_{t+1}$, the clients subset to perform the learning in the next round and build $P(C_{t+1})$. Finally, it sends a transaction $\nrtt$ (New Round) with  $\mathcal{P}(G_t)_{cand}$, $C_{t+1}$ and $P(C_{t+1})$;
    
    \item When a server $i$ sees a $\nrtt$ transaction on-chain for the first time for a round $t$, it retrieves the missing updates and the previous model (if necessary) from the storage;

    \item \label{loop:last-step} When a server $p$ has all the information, it computes the aggregated model $w_t^p$ (by performing robust aggregation to Byzantine clients as defined in Section~\ref{par:fedavg}), sends it to the other servers\footnote{It suffices having only $f_s+1$ selected servers that diffuse their aggregated model to have at least one $\epsilon$-valid model that obtains a \poia. We omit that to keep the code simple.} to build a PoA$\&$I for it. Finally, it sends the PoA$\&$I for $w_t^p$ to clients (which is used in steps 4 to retrieve the model). Round $t$ ends and we move to $t+1$.
    %\item  Each selected client for round $t+1$ performs the learning to produce an update, signs it, and sends it directly to all servers;
    
    \item The workflow loops between steps \ref{loop:first-step} and \ref{loop:last-step} (a round) as long as the learning does not halt (the halt condition is set in $param_{fl_i}$). At the end of the learning, each server produces information to assess quantitatively clients' contributions to the $fl_i$ task, builds a PoA$\&$I and sends a transaction with those information. This transaction allows clients to be automatically rewarded and the model owner to retrieve the model from the storage.   
\end{enumerate}
Interestingly, as we detailed, for each round $t\geq 1$ the blockchain registers a single transaction per round.\footnote{In the practice there is one transaction per server, but processes wait only for the first valid one on-chain to proceed with their computation.}
Fig. \ref{fig:sequence diagram} details the workflow, involving replicated servers, storage replicas, a blockchain, clients, and a model owner. In particular, we use a replicated activation box image on the server lifeline, corresponding to the {\sf build PoA$\&$I} operation. With some abuse of notation, this represents that this operation involves all servers together.

\begin{figure*}[!t]
    \centering
    \includegraphics[width=.95\linewidth]{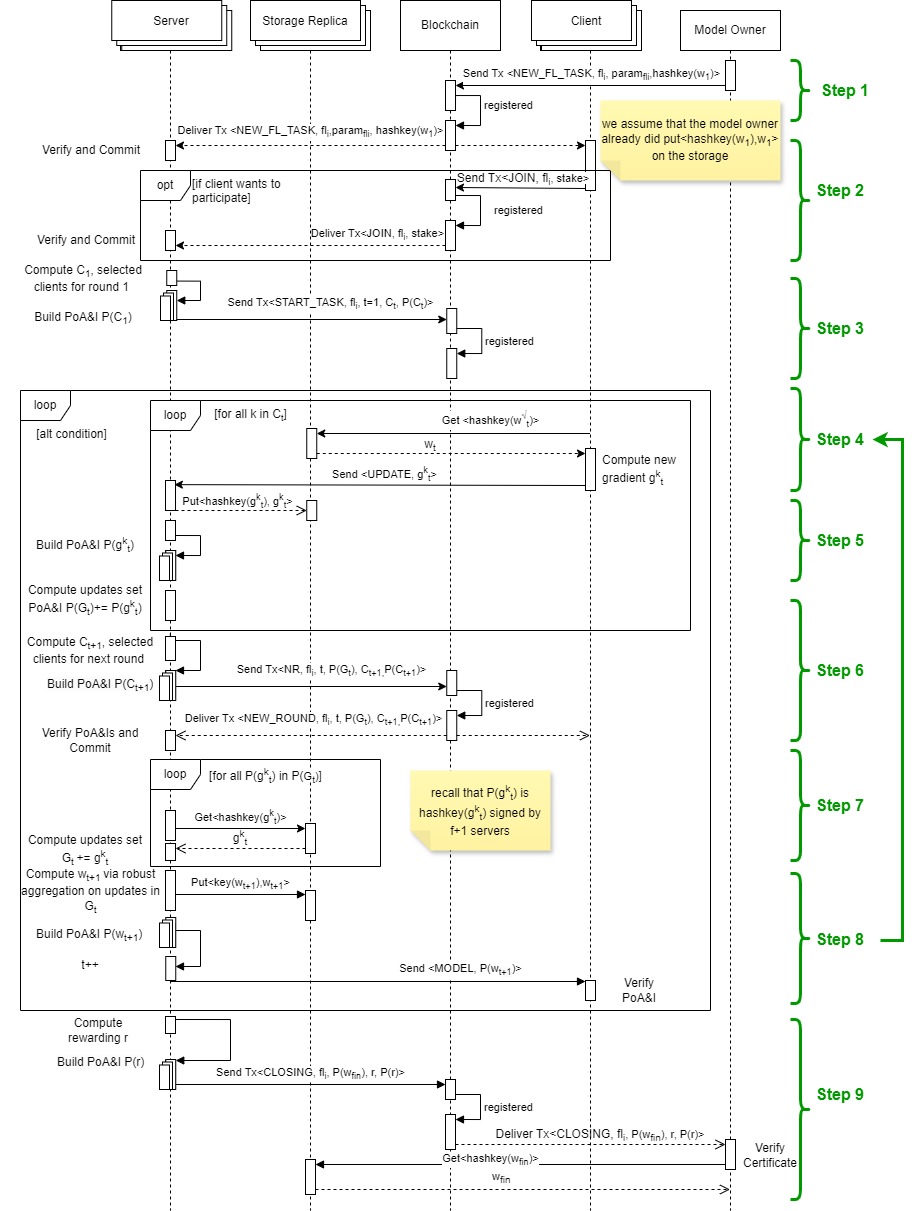}
    \caption{Sequence diagram of the Fantastyc workflow.}
    \label{fig:sequence diagram}
\end{figure*}

%the information recorded on-chain per round. At round $t=0$, the blockchain registers the initial transaction issued by the model owner to set $fl_i$, and the transactions sent by clients willing to join the computation. For each round $t\geq 1$ the blockchain registers a single $\nrtt$ transaction to fix the update set for the round $t$ and the selected clients for round $t+1$ \footnote{In the practice there is more than one $\nrtt$ transaction, but processes wait only for the first valid one before proceeding with their computation.}. At the end of the learning, a last transaction ends the computation and rewards clients.

%\adp{Stress that with PoA$\&$I the computation is done outside consensus}

\begin{comment}

\begin{figure*}[]
\centering
%\centerline{
  %\includegraphics[width=\textwidth]{images/FantaSpec.jpg}
  \includegraphics[scale=0.19]{images/FantaSpec.jpg}
%}
\caption{Sequence diagram describing a federated learning process for a single task}
\label{fig:sequenceDiagram}
\end{figure*}

\end{comment}

\subsubsection{{\sf BuildPoA$\&$I} - Computational Integrity check}\label{par:povs}

The \poia production applies to the client set $C_t$, client $k$'s update $g_t^k$, aggregated model $w_t$, and rewarding information $r$. Each of these elements undergoes a {\sf locally\_valid} check (cf. Algorithm~\ref{algo:server}).
The specific validation process depends on their type. When handling the update $\langle \updtt, g_t^k\rangle$, the check involves verifying the correct formatting of $g_t^k$ and ensuring that $k \in C_t$.
In client set case, as $C_t$ results from the same deterministic computation executed by all correct servers, there is no need for additional verification.
The same applies for rewarding information $r$. However, the aggregated model, resulting from a non-deterministic computation, may slightly differ across correct servers. Each server $p$ potentially produces a different $w_t^p$. As formalized in the next section, we allow for an $\epsilon$ gap between aggregated models. Thus, $p$ considers another model $w_t^q$ valid if $\epsilon$ close to $w_t^q$.

\subsubsection{Floating points and non-determinism handling}\label{sec:non-det}

While servers perform computations using the same update set, slight variations arise due to the non-determinism of floating-point representation. This discrepancy poses a challenge when a server $i$ has to validate the model $w_j^t$ produced by server $j$ having its own model $w_i^t$ as reference.

On a positive note, the disparity between computed models is small enough (cf. note~\ref{footnote:floating}) to introduce an acceptance vector gap that we denote as $\epsilon \in \Re^d$. We assume that $\epsilon$ is static and provided by the model owner when specifying $fl_i$. We assume that $\epsilon$ is sufficiently large to prevent a deadlock: a round in which there are no models issued by servers that obtain a \poia (if their difference is bigger than $\epsilon$).
%\adp{to discuss in the conclusions how to make it dynamic}

\begin{assumption}\label{ass:epsilonmodels}
    Let $\mathcal{W}_{fl_i,t}^C$ be the set of all aggregated models produced by correct servers in a round $t$ for the task $fl_i$ via a robust aggregation function as defined in Section~\ref{par:fedavg}. Let $\epsilon \in \Re^d$ be the acceptance vector gap for the $fl_i$ task. We assume that each pair of models in $\mathcal{W}_{fl_i,t}^C$ is $\epsilon$-close.
\end{assumption}
We say that an aggregated model $w^t$ is $\epsilon$-valid if it is $\epsilon$-close to at least one model in $\mathcal{W}_{fl_i,t}^C$.

\subsection{Correctness arguments}\label{sec:proofs}

In the next we consider the following Assumption:
\begin{assumption}\label{ass:selectedClients}
    The set of clients that join an FL task is such that {\sf select\_clients} returns a set that contains less than half Byzantine clients, enough to produce sufficiently many updates at each round to compute the aggregation robust to Byzantine clients (as defined in Section~\ref{par:fedavg}).
\end{assumption}

With this assumption, we abstract the implementation of the selection function, assuming only that the selection is verifiable. In practice, the selection function can be secret to ensure resilience against an adaptive adversary that can choose which clients to corrupt, or it can be public if only a static adversary is assumed.\\

\subsubsection{Safety}

Let $C^C_t$ and $S^C$ be the set of selected correct clients at round $t$ and the set of correct servers respectively. 
\begin{lemma} (availability and integrity)\label{lem:poia}
    For any $({\tagtt}, v)$, if exists a valid $P({\tagtt},v)$, then ${{\sf locally\_valid}}({\tagtt}, v)$ holds for at least one correct server, and eventually, $v$ is available from the distributed storage. 
\end{lemma}

\begin{proof}[Sketch]
    In Algorithm \ref{algo:read}, a process invokes a {\sf read} operation either in {\sf GET} ({line \ref{algo-read-get}}) or in {\sf GETALL} ({line \ref{algo-read-getall}}) if the input is a \poia.
    By definition, a $P({\tagtt},v)$ contains $f_s+1$ different 
    signatures on the same hash-key $h(v)$. Hence, at least one correct server $p$ signed.
    $p$ executes line 27, Algorithm~\ref{algo:server}, if ${\sf locally\_valid}({\tagtt},v)$ holds and after 
    having sent $\langle h(v),v \rangle$, to all storage replicas (lines 22 and 24). 
    (consequently, ${{\sf locally\_valid}}({\tagtt}, v)$ holds for at least one correct server among the servers that signed).
    Eventually, all correct replicas receive $\langle h(v),v \rangle$. As the hash of $v$ corresponds to the hash-key $h(v)$, correct replicas store it forever (it is not possible to modify it, e.g. sending $\langle h(v),v' \rangle$). 
    Let $r$ be the first correct replica that stores $\langle h(v),v \rangle$ and let $t$ be this time.
   
    Let $q$ be a correct process that given $P({\tagtt},v)$ queries the key-value store for $h(v)$ after $t$. $q$ executes line~\ref{algo-read-get}, Algorithm~\ref{algo:read}. After verifying the proof validity (line~\ref{algo-read-check}), $q$ invokes a {\sf read} operation over the storage, sending a QUERY message to all replicas (line~\ref{algo-read-request}) and returns as soon as he gets a value, whose hash corresponds to the requested one. We need to prove that the {\sf read} terminates.
    By assumption, there are $f_r+1$ replicas. 
    As line~\ref{algo-read-request} occurs after $t$, then at least $r$ receives the query message and replies with $v$ to $q$. 
    The {\sf read} returns as $p$ receives the replay $\ell'$ from $r$. As sent by $q$, the hash of $\ell'$ matches with the requested hash-key, so the {\sf read} returns $v$ at line \ref{algo-read-return}.%, which concludes the proof.
\end{proof}

\begin{theorem}(Safety)
    If a model ownerobtains a $w_{fin}$ then $w_{fin}$ is $\epsilon$-valid %(i.e. $\epsilon$-close to a model produced by a correct server using a robust aggregation function as defined in Section~\ref{par:fedavg}) 
    and issued from a sequence of only $\epsilon$-valid  models.
\end{theorem}
%\adp{add that the aggregation function is robust to byzantine clients. to be placed with assumption 2.}
\begin{proof}[Sketch]
    Each $C_t^C$ client computes an update from an aggregated model $w_{t-1}$ if $w_{t-1}$ possesses a PoA$\&$I (cf. step 4). Hence $w_{t-1}$ is $\epsilon$-valid (Lemma \ref{lem:poia}). 
    The proof concludes by extending the reasoning to the entire sequence of models produced during the FL task as each associated to a \poia.
\end{proof}

\subsubsection{Liveness}

\begin{lemma}\label{lem:poialiveness}
    For each $t\in [0, t_{fin}]$ (steps \ref{loop:first-step}-\ref{loop:last-step}), 
    each correct server has sufficiently many PoA$\&$I for ${(\tagtt, v)}, \tagtt\in \{\updtt, \modtt, \clientstt\}$ to terminate the round.
    %\footnote{Let us recall that the $REWARD$ messages are not specified in the pseudo-code, as issued just once at the end of the FL task, those are handled similarly to $\clientstt$ messages.}.
\end{lemma}

\begin{proof}[Sketch]
We start from round $t=0$: %, and proceed by construction:

\noindent{\bf-} $\updtt$: 
    At $t=0$, there is a sufficient number of clients in $C_t^C$ to produce the required updates for the next robust aggregation (Assumption~\ref{ass:selectedClients}). 
    By recording the hash-key of the initial model on-chain, the Model Owner (MO) allows each $k\in C_{t}^C$ to retrieve the model from storage, compute, and send the update $g_{t=0}^k$ to the servers.
    All $S^C$ servers, eventually receive such updates, which pass the ${\sf locally\_valid}$ check.
    Thus, all $p \in S^C$ send the same message $\langle(\lptt, \tagtt, h(g_{t=0}^k)_p)\rangle$ (modulo the signature) to all servers. Each $S^C$ server
    collects $f_s+1$ $(\langle\lptt, \tagtt, h(g_{t=0}^k)_p\rangle)$ and produces a \poia for each update issued by each $k\in C_{t}^C$.

\noindent{\bf-} $\clientstt$: The same reasoning as in the previous case applies. Indeed, $C_t$ is the result of the same deterministic computation at all $S^C$ servers (line 13) which all send $\langle \lptt, \clientstt, h(C_t)_p\rangle$ to all servers (line 28).

\noindent{\bf-} $\modtt$: From the $\updtt$ case and the Assumption~\ref{ass:selectedClients} it follows that eventually line \ref{algo-server-candidateset} Algorithm~\ref{algo:server} holds at each $p \in S^C$.
   $p$ executes line 13 to select the next client set and starts building the \poia for it (line 14).
    Given the $\clientstt$ case, eventually lines 35 and 39 hold for $\tagtt=\clientstt$, and $p$ executes line 40 sending a $\nrtt$ transaction on the blockchain carrying a candidate update set and the next client set.
    The first of those $\nrtt$ transactions on-chain triggers the execution of line 15 at each $p \in S^C$.
    $p$ retrieves the missing updates (if any) using  their \poia (line 19). $p$ computes the new aggregated model $w_p^{t}$, and builds a \poia for it.
    This implies that $w_p^t$ is sent to all other servers (cf. line 27).
    Upon its reception (line 29), line 30 holds at each correct server for Assumption~\ref{ass:epsilonmodels}. Thus, all $q \in S^C$s send $\langle \lptt, \modtt, h(w_p^t)_q\rangle$ (modulo the signature) to all servers.
    All $S^C$ servers receive $f_s+1$ messages to build a PoA$\&$I for $w_p^t$, which triggers the sending of the \poia model to clients and the computation to move to the next round (lines 43-44).

The difference for $t\in [1, t_{fin}]$ is from where clients get the model hash-key. From $t \geq 1$, the hash-key is provided by servers via its \poia (line 43). Once $C_t^C$ clients retrieve one of the $\epsilon$-valid $w_p^t$ the process proceeds as described for $t=0$. 
\end{proof}

\begin{theorem}(Liveness)
     If a model owner sets up an FL task, then eventually the model owner obtains a trained model.
\end{theorem}

\begin{proof}[Sketch]
   Through the iterative application of Lemma~\ref{lem:poialiveness} to each round, it is inferred that when the model ownerstarts an FL task, this progresses round after round. %We need to prove that the MO receives the \poia of the final trained model.
    As discussed in step \ref{loop:last-step} and in Algorithm~\ref{algo:server} at line 41, each server $p\in S^C$ knows when the training terminates. 
    Hence, $p$ produces a \poia for the rewarding information (similarly to the $\clientstt$ case in Lemma~\ref{lem:poialiveness}) and for the last model ($t_{fin}$ in Lemma~\ref{lem:poialiveness}).
    $p$ has all the elements to send the final transaction, which is eventually on-chain and accessible to the MO.
\end{proof}

\section{Experimental Evaluation}\label{sec:eval}

This section evaluates the practicality of our solution in a realistic federated learning scenario. Initially, we characterize the required number of clients to be selected in each round of the FL process to achieve good accuracy (Section \ref{sec.experimental-model}). Subsequently, we assess the accuracy loss when implementing confidentiality and robust aggregation methods (Section \ref{subsec:robustness}). Both evaluations are conducted in a scenario with only one server and co-located clients. 
Following that, we delve into the effects of server decentralization in a geo-replicated deployment (Section \ref{subsec:experiments distributed system set up}). %-- which preserve robustness achieved by the robust aggregation function, under the assumption of a majority of honest servers (as shown in the correctness proofs).  
We evaluate communication overheads introduced by our decentralization approach, specifically analyzing the contributions of the blockchain and the \poia to round latency. As a baseline, we measure the round latency with one server and clients spread across regions on the Internet. This baseline provides the performance benchmark for a classical FL process with multiple clients and a single server -- a non Byzantine fault-tolerant configuration with $f_s=0$. We will study the latency overhead due to different fault tolerance thresholds.

It is worth noting that our evaluation does not deal with the demonstration of high throughput because we focus on only one FL task. Let us also stress that FL operates in rounds that can be lengthy, not imposing strict latency constraints. Furthermore, we do not compare our approach with other blockchain-based methods due significant differences in architectures, use and choice of blockchains, along with the  impracticality of storing large models in the blockchain (see Section \ref{sec:stateoftheart} for a discussion on related work)--  an issue that limits realistic implementations. 

\subsection{System implementation and deployment}
We implemented our blockchain-based solution using Tendermint\footnote{https://docs.tendermint.com/v0.34/} (now CometBFT) \cite{tendermint} as blockchain consensus  and for the implementation of our servers we use the ABCI for decoupling the application logic from the blockchain (see Section~\ref{sec:prelim-blockchain}). All the algorithms presented in Section \ref{sec:workflow} are implemented as part of the application logic. The FL clients and the model owner communicate directly with the blockchain as Tendermint clients. We use the pub/sub/query protocol Zenoh\footnote{https://zenoh.io/} to implement the fault-tolerant key-value store. 
The application logic at the servers is implemented in Python\footnote{https://www.python.org/} and the FL training at the clients is based on the TensorFlow library\footnote{https://www.tensorflow.org/}.
%The application logic at the servers is implemented in Python\footnote{https://www.python.org/} and the FL training and aggregation (respectively at clients and servers) uses the TensorFlow library\footnote{https://www.tensorflow.org/}.

\subsection{Experimental setup - model and dataset}
\label{sec.experimental-model}
We use the Federated Extended MNIST dataset\footnote{https://www.nist.gov/itl/products-and-
services/emnist-dataset} dedicated to the handwritten digit classification task as experimental setup. This MNIST version contains 10 classes and comes with author id in such a way that its federated version was built by partitioning the data based on the writer. The dataset contains 3,579 writers (each with train/test data split for a total of 240,000 train and 40,000 test images) that permits to test the robustness and scalability of our solution. Evaluations were done with a standard CNN (the same as the one used in \cite{McMahan2017}) composed of two $5*5$ convolution layers (the first with 32 channels, the second with 64 channels, each followed with $2*2$ max pooling), a fully connected layer with 512 units and ReLu activation, and a final softmax output layer (1,663,370 total parameters). For local training, the number of epochs and the minibatch size were set to $E=10$ and $B=5$ respectively.

As benchmark, we report the accuracy of the federated model on the test set throughout the training process without malicious clients or countermeasures to prevent information leakage. We experiment with the number $K$ of participants per round, which controls the amount of multi-client parallelism (Fig.\ref{Fig.FLparallelism}.a). We show the impact of varying $K$ about the speedup of the learning process and it appears that there is no more advantage in increasing the client fraction above $K=50$ (Fig.\ref{Fig.FLparallelism}.b). 
%Thus, unless otherwise stated, 
Thus, for further experiments we fix $K=50$, which strikes a good balance between computational efficiency and convergence rate.
%Thus, for our experiments about robustness, we fix $K=50$, which strikes a good balance between computational efficiency and convergence rate. However, for our experiments about scalability, we vary $K$ again to evaluate the number of transactions supported by our BC-based FL solution.

\begin{figure}[hbt]
	 \centering	\includegraphics[width=1\linewidth,keepaspectratio=true]{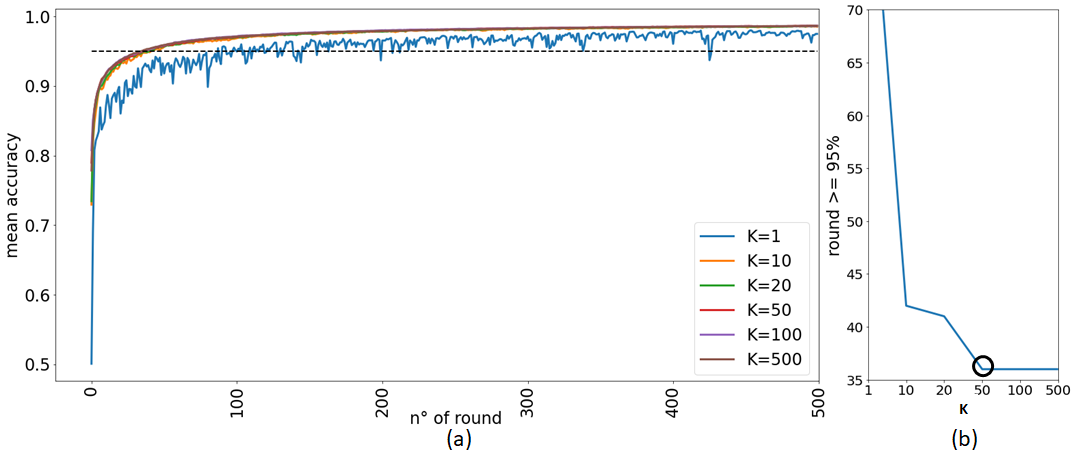}\vspace{-0.2cm}
	    \caption{Study about the impact of multi-client parallelism varying $K$ (while setting $B=5$ ; $E=10$). (a) Test set accuracy vs. communication rounds varying the number of participants (b) Number of communication rounds to reach a target accuracy of 95$\%$ depending on the number of participants per round.}
   \label{Fig.FLparallelism} 
\vspace{-0.2cm}
\end{figure}

\subsection{Privacy and robustness trade-offs evaluation}\label{subsec:robustness}

We have conducted experiments to evaluate accuracy loss due to the addition of our modules dedicated to privacy and integrity issues (Table \ref{tab:datasets}).  Without
malicious clients or countermeasures to prevent information leakage, the use of FedAvg as aggregation rule permits to reach the best accuracy of 98.6\% after 500 rounds of training. We illustrate next the necessary trade-off between privacy, integrity and accuracy.

The risk of privacy leakage decreases as the number of samples $s$ used to build InstaHide composite images increases but in the same time, the training process suffers small accuracy loss of 1.1\% to 3.7\% when s varies from 2 to 5.

To evaluate the robustness of our solution to Byzantine attacks, we consider that there are $f_c = \{ 10, 20,30 \}$ malicious clients among 50 participants at each round of training. By only sending $-\lambda_{boost} g_t^{k}$ instead of $g_t^{k}$ (where $\lambda_{boost}=5$), it is very easy for Byzantine clients to impeded a training process based on FedAvg as the model comes back to random with 10\% accuracy. The use of a Byzantine-tolerant aggregation rule such as the median is thus essential to withstand Byzantine attacks and assure the converge of the model -- as long as the number of honest participants remains a majority. % K=50 89.7 ; K=100 91.4 ; K=500 91.3

\renewcommand{\arraystretch}{1.4}
\begin{table}[h]
    \centering
    \begin{tabular}{|p{0.8cm}||p{0.2cm}p{0.2cm}p{0.2cm}p{0.2cm}p{0.35cm}|p{0.2cm}p{0.2cm}p{0.35cm}|p{0.2cm}p{0.2cm}p{0.35cm}|}
    \hline
         $f_c$ & \multicolumn{5}{|c|}{0} & 10 & 20 & 30 & 10 & 20 & 30 \\
        \hline
         Mix & no & 2 & 3 & 4 & 5 & \multicolumn{3}{|c|}{no} & \multicolumn{3}{|c|}{4}\\
        \hline
        \hline
         FedAvg & 98.6 & 97.5 & 96.5 & 95.6 & 94.9 & 10.0 & 10.0 & 10.0 & 10.0 & 10.0 & 10.0  \\
         \hline
         median & 96.7 & 95.2 & 93.8 & 92.0 & 90.5 & 96.3 & 94.7 & 10.0 & 90.9 & 89.7 & 10.0  \\
         \hline
    \end{tabular}
    \vspace{1em}
    \caption{Model accuracy after 500 rounds of training with $K=50$ clients and considering {FedAvg;median} as aggregation rule. We deploy InstaHide with $s=\{2;3;4;5\}$ mix samples and simulate $f_c=\{0,10,20,30\}$ byzantine clients with $\lambda_{boost}=5$.}
    \label{tab:datasets}
    %\vspace{-0.4cm}
\end{table}

\subsection{Experimental setup - distributed system}\label{subsec:experiments distributed system set up}

We run our experiments in the cloud computing service Akamai-Linode\footnote{https://www.linode.com/}, which offers dedicated CPU and storage resources that can be allocated in geographically distant regions. 
For the experiments, we evaluate a deployment where within each node we collocate a server, a blockchain validator, and a storage replica (a peer, in Zenoh's parlance). We consider several configurations where the nodes and the clients are spread in up to four different regions (EU, US-west, US-east, Asia). At each region we collocate one or more instances (dedicated Linode machines with 32GB of memory) which host the nodes. For the FL clients, we let each Tendermint client (called a proxy client) to generate the workload of a set of realistic participants (recall a participant is an FL client that takes part in the training). At each round, each proxy client sends the payloads from the set of participants it contains in sequential fashion. We deploy up to two instances per region, each of them with a number of nodes that varies between one and 16. We evenly distribute the clients among instances deployed in all regions, where each instance may contain up to 8 proxy clients. Each proxy client is capable of generating the workload for between 2 to 5 participants. Typically, we ensure a balanced distribution by limiting the number of instances to a maximum of four and adjusting the number of participants per proxy client to fall within the range of 10, 50, or 100 participants in total. The communication bandwidth between the regions is in the range of 3Mbps--10Mbps.

\subsection{Scalability evaluation}
We first study the latency associated to a single round and then we show the communication overhead within the overall FL process, with varying number of regions, nodes, and clients.\\
\noindent\textbf{Round analysis.}
Round latency includes the block time (time to add in the blockchain a block that includes the transaction sent in the round), the PoA\&I time (to generate a valid PoA\&I), as well as the contribution to complete steps~\ref{loop:first-step}--\ref{loop:last-step} in the workflow of Section~\ref{sec:workflow}. Importantly, at each round there is only one write access to the blockchain (step~6).

%Preliminary experiments with different payload sizes (up to 12MB, which corresponds to the size of the model considered in Section~\ref{sec.experimental-model} above) confirm that the block time is independent of the payload size, which is to be expected given Section~\ref{subsec:efficient-use}. On the contrary, the numbers of nodes and regions have an impact on the number of messages exchanged in order to reach consensus. Further experiments with fixed numbers of participants yield a practical limit of 16 nodes per region, above of said limit the block time latency degrades very quickly due to depletion of resources at the Linode instances. Within this limit, configurations of up to 16 nodes evenly spread over four instances show that the block time remains always in the range of 1--5s. Our typical configuration considers four instances evenly deployed over two or four regions, each of them with up to four nodes. %With respect to clients, preliminary experiments yield a practical limit of 50 participants for instance. Within this limit, configurations with five proxy clients and 10 participants per proxy client are generally the most favourable.

\textit{Blocktime contribution.} Fig.\ref{Fig.Scalability} plots latency results with $K=50$ participants per round and a varying number of nodes, training the CCN model considered in Section~\ref{sec.experimental-model}, whose payload corresponds to 12MB.  Preliminary experiments yield a practical limit of 16 nodes, above of said limit the latency degrades very quickly due to depletion of resources at the Linode instances. The top two curves correspond to round latency with payload 12MB, with either four regions and one instance per region (blue) or two regions and two instances per region (orange). The best round latency with 16 nodes is 96s with two regions. The next two curves show the round time respectively with a mocked model of size 6KB. The best latency with 16 nodes is 20s with four regions. %The speed up with respect to the two curves above is due to the forwarding of a much smaller payload by the storage replicas. 
Finally, the base line plots the block times associated to the configurations of the four curves above (green and red for 12MB payload and four or two regions respectively, purple and brown for 6KB payload and four or two regions respectively).
These experiments confirm that the block time is independent of the payload size, which is to be expected given Section~\ref{subsec:efficient-use} and smoothly increases with the number of nodes in nominal network conditions. 

\begin{figure}[hbt]
	 \centering	\includegraphics[width=0.98\linewidth,keepaspectratio=true]{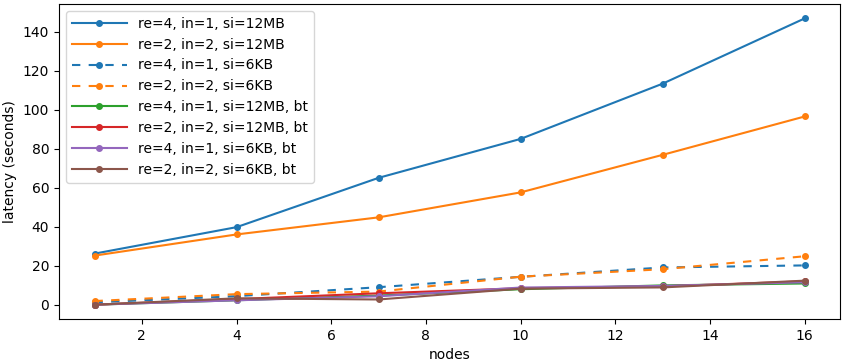}
        \vspace{-0.3cm}
	    \caption{Evaluation of the impact of the number of nodes with a constant number of 50 clients.}
   \label{Fig.Scalability} 
\end{figure}

\textit{\poia contribution.}
Fig.\ref{Fig.PoAI} presents the latency breakdown across different intervals of the learning round, maintaining the same client configuration as described earlier for 4, 10, 16 nodes, respectively, evenly deployed over 2 regions, each with 2 instances.  In the plot, the time for the client to gather the model from the storage (computed in the previous round) and verify model's \poia (step 4) is represented in blue, the time for servers to collect client updates sent by clients and generate the PoA\&I (step 5) is in orange, the block time (step 6) is in green, the time for servers to verify updates'  \poia and get the last model from the storage (step 7) is in red, and the time to aggregate updates into the new model and generate the PoA\&I of the aggregated model (step 8) is in purple. 

\begin{figure}[hbt]
	 \centering	\includegraphics[width=0.98\linewidth,keepaspectratio=true]{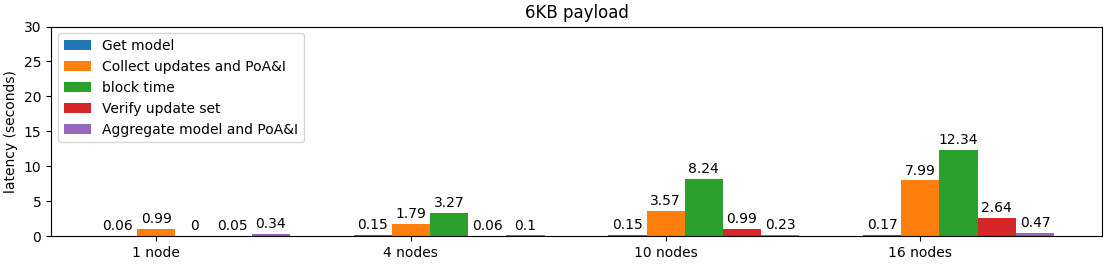}
     \includegraphics[width=0.98\linewidth,keepaspectratio=true]{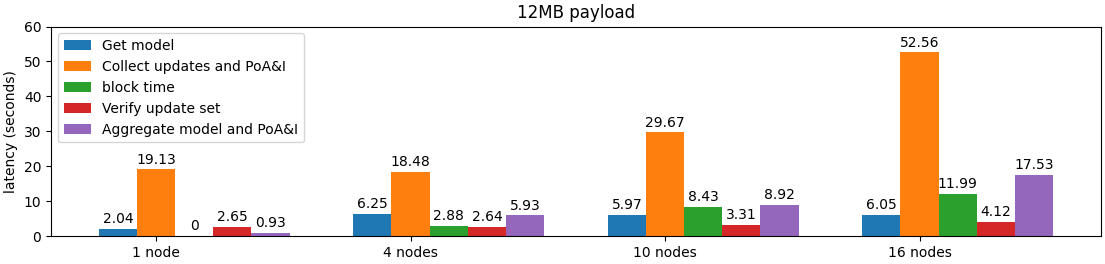}
        \vspace{-0.2cm}
	    \caption{Breakdown of latency within a round with 50 clients.}
   \label{Fig.PoAI} 
\end{figure}

%Observing the results, we find that in the case of small models, blocktime emerges as the primary factor influencing latency ($3s$ for 4 nodes and $12s$ for 16 nodes), introducing a substantial overhead with respect to the baseline ($3s$) with no Byzantine tolerance, i.e. $f_s=0$. However, for larger models, the dominant contribution stems from the time required to transfer updates from the $50$ clients to servers (orange bar). %With $f_s=0$, this time is equal to $19,3s$. With $4$ nodes this time does not change significantly (due to a load balancing effect, where some clients are closer to some replicas), while for $10$ nodes it increases to $29,67s$.  

As for the \poia contribution to latency for client updates (included in the orange bar),  we can see that the \poia generation for client updates is quick; witnessed by the comparable latency between the baseline (without \poia) and the 4-node configuration. This is due to a voting process involving small-sized votes (less than 1KB) and the fact that each \poia production is run asynchronously as soon as a client update is received by the server, and this for each update received in the round.

The purple bar shows the time spent to generate the \poia for the model. For this \poia, each node publishes the model for the other server nodes. The only linear increase in time with the number of nodes is because the latency is dominated by the publishing of the model, while the latency to get $f_s+1$ votes is negligible, due to the extremely small size of quickly published vote messages. As for the blue bar, the client only needs to retrieve one model with a valid \poia. This explains the constant contribution of the step 4 no matter of the number of nodes. Without a \poia, each client would be required to wait for at least $f_s+1$ computed models,  which would have a prohibitive cost with a large number of client participants.  

%These findings emphasize the costs associated with decentralization and security, primarily due to the degree of replication in the server set, which are unavoidable for Byzantine fault-tolerance. Nicely, our solution shows  manageable costs with a small set of servers and large models, by decoupling the blockchain cost from the message size and by adding a negligible overhead for the \poia. 

\noindent \textbf{Communication overheads.} For a comprehensive evaluation of the results, Table \ref{tab:faulttoleranceoverhead} illustrates the communication overhead for various fault-tolerance thresholds $f_s$, for a number of participants per round $K=10$, $K=50$, $K=100$ (which implies $f_c$=4, $f_c$=20 and $f_c$=40). 
We can consider the threshold for servers, such that $n_s=2f_s+1$, but since in the considered deployment the server is collocated with a blockchain validator in the same node, we can also fairly consider $n_s=3f_s+1$ -- this depends on granularity of the failure model, at node level or at software component level. We consider both, distinguishing them within the interval  [$f_s^{min},f_s^{max}$].  To provide a clearer understanding of the overhead, we present the total time \textbf{T} spent for training calculated as the latency per round multiplied by $500$, and time to reach good accuracy $T_{a}$ ($>95\%$), which is 37 rounds for $K\ge 50$ and 43 rounds for $K=10$ (Section \ref{sec.experimental-model}). 
From the results it can be noted than on average for the smaller model the overhead increases linearly with the number of nodes. On the other hand, for the large model, the overhead increases sub-linearly with the number of nodes. In general, our approach is able to optimise the case of the larger models.  As observed in the latency breakdown (Fig.\ref{Fig.PoAI}), this is due to the
time for a client to retrieve and verify the new model from its \poia, as it waits for the first good reply (blue bar) making this time constant with respect to the number of nodes,
%constant contribution of the blue bar regardless of the number of servers (the client optimizes the time to retrieve the new model by waiting just one server replica that provides a valid model),
the asynchronous process to generate a \poia (orange bar) leveraging parallelism when the number of servers increases, and an acceptable contribution of the block time experienced in nominal network conditions.

\begin{table}[h]
    \centering
    %\tiny
    \begin{tabular}{|p{2,45cm}||p{1.5cm}p{1.5cm}p{1.7cm}|}
    %\begin{tabular}{|p{1.7cm}||p{.7cm}p{.7cm}p{.7cm}p{.7cm}|}
    \hline
         $n_s$ & 1 (baseline) & 4 & 16 \\
    \hline
         $[f_s^{min},f_s^{max}]$ & [0,0] & [1,1] & [5,7] \\
         \hline
           [\textbf{T},$T_a$] 6KB,K=10 & [\textbf{7}, 0.6]min & [\textbf{12}, 1.1]min & [\textbf{135}, 11.6]min\\
        \hline
          [\textbf{T},$T_a$] 6KB,K=50  & [\textbf{12}, 0.9]min & [\textbf{45}, 3.3]min & [\textbf{197}, 15]min \\
          \hline
         [\textbf{T},$T_a$] 6KB,K=100& [\textbf{17}, 1.3] min & [\textbf{70}, 5.2] min & [\textbf{247}, 18]min\\
         \hline
         \hline
          [\textbf{T},$T_a$] 12MB,K=10 & [\textbf{1}, 0.09] hr & [\textbf{2.4}, 0.2] hr & [\textbf{8.5}, 0.73] hr\\
          \hline
          [\textbf{T},$T_a$] 12MB,K=50  & [\textbf{3.4}, 0.25] hr & [\textbf{5}, 0.4] hr & [\textbf{13}, 1] hr \\
          \hline
          [\textbf{T},$T_a$] 12MB,K=100& [\textbf{6.7}, 0.5] hr & [\textbf{9.4}, 0.7] hr & [\textbf{15.3}, 1.1] hr\\
         \hline
    \end{tabular}
    \caption{Communication overhead with different fault thresholds.}
    \label{tab:faulttoleranceoverhead}
\end{table}

\section{Conclusion and Perspectives}\label{sec:conclusion}

This paper introduced Fantastyc, a blockchain-based federated learning system that is provably secure, improves Byzantine tolerance, limits the risks of sensitive information leakage and showcases practical applicability by amortizing the cost of growing workloads with the number of server replicas in nominal network conditions.  %Our findings highlight that leveraging the blockchain has minimal impact on performance in nominal network conditions, as an off-chain set of servers handles model computation and generates proofs, ensuring the integrity of the FL process at any given time. 
Future work will primarily concentrate on two concurrent axes: the first involves further improving the underlying communication overhead by implementing gradient compression techniques, %(in the same spirit of \cite{kollias2022sketch} in which servers can compare their models before incurring in expensive communication costs)
 studying privacy/utility trade-offs. The second axis revolves around incentive mechanisms that can leverage the participation count of FL contributors in each round to address the issue of free-riding\footnote{Free-riding in this context is related to the problem of an agent that will prefer not endeavoring the cost of learning while presenting fake gradients.}. Indeed, even if completely eliminating free-riding without access to each client's specific dataset remains an unsolved challenge \cite{Fraboni2021},  robust incentive mechanisms encouraging agents to exhibit honest behavior can be designed, assuming the presence of an auditor \cite{Gao2019}. Blockchain-based auditor mechanisms are then a promising avenue of research. %Moreover, thanks to the honest participation incentive, that would permit to lift Assumption \ref{ass:selectedClients} on the number of correct clients.
As part of our future work, we also explore a mechanism for dynamically adjusting $\epsilon$. 
That would allow to lift Assumption \ref{ass:epsilonmodels}, enabling to start the FL process with a small $\epsilon$, without risking a deadlock. 
Finally, we could also make our solution robust to backdoor attacks \cite{Bagdasaryan2020}. As our architecture is modular, it seems straightforward to integrate a post-training defense such as the one presented in \cite{castellon2023}. 

\section*{Acknowledgements}
This work is funded by The Fantastyc Carnot project.

\bibliographystyle{IEEEtran}

\bibliography{bib}
% \appendices
% \section{Solution Sequence Diagram}\label{app:seqDiagram}
% \input{app_sequence.tex}
% \section{Correctness Arguments}\label{app:proofs}
% \input{app_correctness.tex}

\end{document}